\documentclass[12pt]{article}
\usepackage{amsmath,amsthm,amssymb,color,verbatim,graphicx,fullpage,url}
\usepackage{tcolorbox}
\usepackage{hyperref}
\usepackage{cleveref}
\usepackage{enumerate}
\newcommand{\remove}[1]{}
\newtheorem{theorem}{Theorem}[section]
\newtheorem{claim}[theorem]{Claim}
\newtheorem{lemma}[theorem]{Lemma}
\newtheorem{proposition}[theorem]{Proposition}
\newtheorem{definition}[theorem]{Definition}

\newtheorem{obs}[theorem]{Observation}
\newtheorem{open}{Open problem}

\newcommand{\A}{\mathcal{A}}
\newcommand{\R}{\mathbb{R}}
\newcommand{\B}{\mathbb{B}}
\newcommand{\E}{\mathbb{E}}
\newcommand{\N}{\mathbb{N}}
\newcommand{\ip}[1]{\langle{#1}\rangle}

\newcommand{\sign}{\textrm{sign}}
\newcommand{\eps}{\varepsilon}
\newcommand{\vol}{\text{vol}}
\newcommand{\infer}{\mathsf{InferComp}}
\newcommand{\inference}{\mathsf{infer}}

\title{Generalized comparison trees for point-location problems}

\author{
Daniel M. Kane\thanks{Department of Computer Science and Engineering/Department of Mathematics, University of California, San Diego. {\tt dakane@ucsd.edu} Supported by NSF CAREER Award ID 1553288 and a Sloan fellowship.}
\and Shachar Lovett\thanks{Department of Computer Science and Engineering, University of California, San Diego. {\tt slovett@cs.ucsd.edu.} Research supported by NSF CAREER award 1350481, CCF award 1614023 and a Sloan fellowship.}
\and Shay Moran\thanks{Institute for Advanced Study, Princeton. {\tt  shaymoran@ias.edu.} Research supported by the National Science Foundation under agreement No. CCF-1412958 and by the Simons Foundations.}
}

\begin{document}
\maketitle

\begin{abstract}
Let $H$ be an arbitrary family of hyper-planes in $d$-dimensions. We show that the point-location problem for $H$
can be solved by a linear decision tree that only uses a special type of queries called \emph{generalized comparison queries}.
These queries correspond to hyperplanes that can be written as a linear combination of two hyperplanes from $H$;
in particular, if all hyperplanes in $H$ are $k$-sparse then generalized comparisons are $2k$-sparse.
The depth of the obtained linear decision tree is polynomial in $d$ and logarithmic in $|H|$,
which is comparable to previous results in the literature that use general linear queries.

This extends the study of comparison trees from a previous work by the authors [Kane {et al.}, FOCS 2017].
The main benefit is that using generalized comparison queries allows to
overcome limitations that apply for the more restricted type of comparison queries.

Our analysis combines a seminal result of Forster
regarding sets in isotropic position [Forster, JCSS 2002],
the margin-based inference dimension analysis for comparison queries
from [Kane {et al.}, FOCS 2017], and compactness arguments.
\end{abstract}

\newpage
\section{Introduction}

Let $H \subset \R^d$ be a family of $|H|=n$ hyper-planes.
$H$ partitions $\R^d$ into $O(n^d)$ cells.
The \emph{point-location} problem is to decide, given an input point $x \in \R^d$, to which cell it belongs.
That is, to compute the function
\begin{align*}
\A_H(x) := \left(\sign(\ip{x,h}): h \in H \right) \in \{-1,0,1\}^n.
\end{align*}

A well-studied computation model for this problem is a \emph{linear decision tree} (LDT):
this is a ternary decision tree whose input is $x\in\R^d$ and
its internal nodes $v$ make linear/threshold queries of the form $\sign(\ip{x,q})$ for some $q=q(v) \in \R^d$.
The three children of $v$ correspond to the three possible outputs of the query : ``$-$'',``$0$'',``$+$''.
The leaves of the tree are labeled with $\{-1,0,1\}^n$ with correspondence to the cell in the arrangement that contains $x$.
The complexity of a linear decision tree is its depth, which corresponds to the maximal number of linear queries made on any input.

\paragraph{Comparison queries.}
A \emph{comparison decision tree} is a special type of an LDT, in which all queries are of one of two types:
\begin{itemize}
\item Label query: ``$\sign\left(\ip{x,h}\right) =\text{?}$" for $h \in H$.
\item Comparison query: ``$\sign\left(\ip{x,h'-h''}\right) =\text{?}$" for $h',h'' \in H$.
\end{itemize}

In~\cite{kane2017active} it is shown that when $H$ is ``nice"
then there exist comparison decision trees that computed
$\A_H(\cdot)$ and has nearly optimal depth (up to logarithmic factors).
For example, for any $H \subset \{-1,0,1\}^d$
there is a comparison decision tree
with depth $O(d \log d \log \lvert H\rvert)$.
This is off by a $\log d$ factor from the basic information theoretical lower bound of $\Omega(d \log \lvert H\rvert)$.
Moreover, it is shown there that certain niceness conditions are necessary.
Concretely, they give an example of $H \subset \R^3$
such that any comparison decision tree that computes $\A_H(\cdot)$ requires depth $\Omega(\lvert H\rvert)$.
This raises the following natural problem:
can comparison decision trees be generalized in a way
that allows to handle arbitrary point-location problems?

\paragraph{Generalized comparisons.}
This paper addresses the above question by considering generalized comparison queries.
A generalized comparison query allows to re-weight its terms: namely, it is query of the form
$$
\text{``}\sign\left(\ip{x,\alpha h'-\beta h''}\right) =?\text{"}
$$
for $h',h'' \in H$ and some $\alpha,\beta \in \R$.
Note that it may be assumed without loss of generality that $ |\alpha|+|\beta|=1$.
A generalized comparison decision tree, naturally, is a linear decision tree whose internal linear queries are restricted to be generalized comparisons.
Note that generalized comparison queries include as special cases both label queries (setting $\alpha=1,\beta=0$) and comparison queries (setting $\alpha=\beta=1/2$).

Geometrically, generalized comparisons are
1-dimensional in the following sense:
let $q=\alpha h' - \beta h''$, with $\alpha,\beta\geq 0$ then $q$ lies on the interval connecting $h'$ and $-h''$. If $\alpha$ and $\beta$ have different signs, $q$ lies on an interval between some other $\pm h'$ and $\pm h''$. So comparison queries are linear queries that lies on the projective lines intervals spanned by $\{\pm h: h\in H\}$.
In particular, if each $h\in H$ has sparsity at most $k$ (namely, at most $k$ nonzero coordinates)
then each generalized comparison has sparsity at most $2k$.

Our main result is:
\begin{theorem}[Main theorem]
\label{thm:main}
Let $H \subset \R^d$.
Then there exists a generalized comparison decision tree of depth
$O(d^{4} \log d \log |H|)$ that computes $\A_H(x)$ for every input $x\in\R^d$.
\end{theorem}


\paragraph{Why consider generalized comparisons?}
We consider generalized comparisons for a number of reasons:
\begin{itemize}
\item The lower bound against comparison queries in \cite{kane2017active} was achieved by essentially scaling different elements of $H \subset \R^3$
with exponentially different scales. Allowing for re-scaling (which is what generalized comparisons allow to do) solves this problem.
\item Generalized comparisons may be natural from a machine learning perspective,
in particular in the context of active learning.
A common type of queries used in practice it to give a score
to an example (say 1-10), and not just label it as positive (+) or negative (-). Comparing the scores for different examples can be viewed as a ``coarse" type of generalized comparisons.
\item If the set of original hyperplanes $H$ was ``nice", then generalized comparisons maintain some aspects of niceness in the queries performed.
As an example that was already mentioned, if all hyperplanes in $H$ are $k$-sparse then generalized comparisons are $2k$-sparse. This is part of
a more general line of research, studying what types of ``simple queries" are sufficient to obtain efficient active learning algorithms, or equivalently
efficient linear decision trees for point-location problems.
\end{itemize}

\subsection{Proof outline}
Our proof consists of two parts.
First, we focus on the case when $H \subset \R^d$ is in \emph{general position},
namely, every $d$ vectors in it are linearly independent.
Then, we extend the construction to arbitrary $H$.

The second part is fairly abstract and is derived via compactness arguments.
The technical crux lies in the first part:
let $H\subseteq \R^d$ be in general position;
we first construct a randomized generalized comparison decision tree for $H$,
and then derandomize it.
The randomized tree is simple to describe: it proceeds by steps,
where in each step about $d^2$ elements from $H$ are drawn,
labelled, and sorted using generalized comparisons.
Then, it is shown that  the labels of some $1/d$-fraction
of the remaining elements in $H$ are inferred, on average.
The inferred vectors are then removed from $H$
and this step is repeated until all labels in $H$ are inferred.

A central technical challenge lies in the analysis of a single step.
It hinges on a result by Forster~\cite{forster02linear}
that transforms a general-positioned $H$
to an \emph{isotropic-positioned} $H'$ (see formal definition below)
in a way that comparison queries on $H'$ correspond to generalized comparison queries on $H$.
Then, since $H'$ is in isotropic position, it follows that a significant fraction of $H'$
has a large margin with respect to the input $x$.
This allows us to employ a variant of the margin-based inference analysis by~\cite{kane2017active} on $H'$
to derive the desired inference of some $\Omega(\frac{1}{d})$-fraction of
the remaining labels in each step.

The derandomization of the above randomized LDT is achieved by a
double-sampling argument due to~\cite{vapnik1971uniform}. A similar argument
was used in \cite{kane2017active}, however here several new technical challenges
arise, as in each iteration in the above randomized algorithm, we only label a small fraction of the elements on average.

\subsection{Related work}
The point-location problem has been studied since the 1980s, starting from the pioneering work of
Meyer auf der Heide~\cite{meyer1984polynomial}, Meiser~\cite{meiser1993point}, Cardinal et al.~\cite{cardinal2015solving}
and most recently Ezra and Sharir~\cite{ezra2017decision}. This last work, although not formally stated as such,
solves the point-location problem for an arbitrary $H \subset \R^d$ by a linear decision tree whose depth is $O(d^2 \log d \log |H|)$.
However, in order to do so, the linear queries used by the linear decision tree could be arbitrary, even when the original family $H$
is very simple (say $3$-sparse). This is true for all previous works, as they are all based on various geometric partitioning ideas,
which may require the use of quite generic hyperplanes. This should be compared with our results (\Cref{thm:main}). We obtain
a linear decision tree of a bigger depth (by a factor of $d^2$), however the type of linear queries we use remain relatively simple;
e.g., as discussed earlier, they are 1-dimensional and preserve sparseness.

\subsection{Open problems}
Our work addresses a problem raised in \cite{kane2017near}, of whether ``simple queries" can be sufficient
to solve the point-location problem for general hyperplanes $H$, without making any ``niceness" assumptions on $H$.
The solution explored here is to allow for generalized comparisons, which are a $1$-dimensional set of allowed queries.
An intriguing question is whether this is necessary, or whether there are some $0$-dimensional gadgets that
would be sufficient.

In order to formally define the problem, we need the notion of \emph{gadgets}.
A $t$-ary gadget in $\R^d$ is a function $g:(\R^d)^t \to \R^d$. Let $G=\{g_1,\ldots,g_r\}$
be a finite collection of gadgets in $\R^d$. Given a set of hyperplanes $H \subset \R^d$,
a $G$-LDT that solves $\A_H(\cdot)$ is a LDT where any linear query is of the form $\sign(\ip{q,\cdot})$ for
$q=g(h_1,\ldots,h_t)$ for some $g \in G$ and $h_1,\ldots,h_t \in H$. For example,
a comparison decision tree corresponds to the gadgets $g_1(h)=h$ (label queries) and $g_2(h_1,h_2) = h_1-h_2$ (comparison queries).
A generalized comparison decision tree corresponds to the $1$-dimensional (infinite) family of gadgets
$\{g_{\alpha}(h_1,h_2)=\alpha h_1 - (1-\alpha) h_2: \alpha \in [0,1]\}$.
It was shown in \cite{kane2017active} that comparison decision trees are sufficient to efficiently solve
the point-location problem in 2 dimensions, but not in 3 dimensions. So, the problem is already open in $\R^3$.

\begin{open}
\label{open:gadget}
Fix $d \ge 3$. Is there a finite set of gadgets $G$ in $\R^d$, such that
for every $H \subset \R^d$ there exists a $G$-LDT which computes $\A_H(\cdot)$, whose depth is logarithmic in $|H|$?
Can one hope to get to the information theoretic lower bound, namely to $O(d \log |H|)$?
\end{open}

Another open problem is whether randomized LDT can always be derandomized, without losing too much in the depth.
To recall, a randomized (zero-error) LDT is a distribution over (deterministic) LDTs which each computes $\A_H(\cdot)$. The measure
of complexity for a randomized LDT is the expected number of queries performed, for the worst-case input $x$.
The derandomization technique we apply in this work (see \Cref{lem:genposdet} and its proof for details) loses a factor of $d$, but it is not clear whether this
loss is necessary.

\begin{open}
\label{open:derand}
Let $H \subset \R^d$. Assume that there exists a randomized LDT which computes $\A_H(\cdot)$, whose expected query complexity
is at most $D$ for any input. Does there always exist a (deterministic) LDT which computes $\A_H(\cdot)$, whose depth is $O(D)$?
\end{open}

\section{Preliminaries and some basic technical lemmas}

\subsection{Linear decision trees}
Let  $T$ be a linear decision tree defined on input points $x\in \R^d$.
For a vertex $u$ of $T$ let $C(u)$ denote the set of inputs $x$ whose computation path
contains $u$. Let  $Z(u)$ denote the queries ``$\sign(\ip{x,q})=?$'' on the path from the root to $u$
that are replied by ``$0$'', and let $V(u)\leq \R^d$ denote the subspace
$\left\{x : \langle x, q\rangle = 0, \forall q\in Z(u) \right\}$.
We say that $u$ is full dimensional if $\dim\left(V(u)\right)=d$
(i.e.\ no query on the path towards $u$ is replied by a $0$).
\begin{obs}\label{obs:convex}
$C(u)$ is convex (as an intersection of open halfspaces and hyperplanes).
\end{obs}

\begin{obs}\label{obs:open}
$C(u)\subseteq V(u)$ and is open with respect to $V(u)$
(that is, it is the intersection of an open set in $\R^d$ with $V(u)$).
\end{obs}

We say that $T$ \emph{computes} $\sign(\ip{h,\cdot})$ if for every leaf $\ell$ of $T$,
the restriction of the function $x \mapsto \sign(\ip{h,x})$ to $C(\ell)$ is constant.
Thus, $T$ computes $\A_H(\cdot)$ if and only if it computes $\sign(\ip{h,\cdot})$ for all $h\in H$.
We say that $T$ \emph{computes} $\sign(\ip{h,\cdot})$ \emph{almost everywhere} if the restriction of
$x \mapsto \sign(\ip{h,x})$ to~$C(\ell)$ is constant, for every {full dimensional} leaf $\ell$.

We will use the following corollary of Observations \ref{obs:convex} and \ref{obs:open}. It shows that if
$\sign(h,\cdot)$ is not constant in $C(u)$ then it must take all three possible values.
In \Cref{sec:proof4}, we show that a linear decision tree that computes $\A_H(\cdot)$ almost everywhere can be
``exteneded" to a LDT that computes $\A_H(\cdot)$ everywhere, without increasing the depth or introducing new queries.
It relies on the following lemma.

\begin{lemma}\label{lem:all}
Let $u$ be a vertex in $T$, and assume that the restriction of $x\mapsto\sign(\ip{h,x})$
to $C(u)$ is not constant. Then there exist $x_{-1},x_0,x_{+1}\in C(u)$
such that $\sign(\ip{h,x_i}) = i$ for every $i\in\{-1,0,+1\}$.
\end{lemma}

\begin{proof}
Let $x',x''\in C(u)$ with $\sign(\ip{h,x'})\neq\sign(\ip{h,x''})$.
If $\{\sign(\ip{h,x'}),\sign(\ip{h,x''})\} = \{\pm1\}$
then by continuity of $x\mapsto \sign(\ip{h,x})$ there exists some
$x_0$ on the interval between $x',x''$ such that $\sign(\ip{h,x_0})=0$,
and $x_0\in C(u)$ by convexity.

Else, without loss of generality, $\sign(\ip{h,x'})=0$ and $\sign(\ip{h,x''}) = +1$.
Therefore, since $C(u)$ is open relative to $V(u)$:
\[x'-\eps\cdot x''\in C(u)\]
for some small $\eps>0$.
This finishes the proof since $\sign(\ip{h,x' - \eps\cdot x'' })=-1$.
\end{proof}

\subsection{Inferring from comparisons}
{Let $x,h\in \R^d$ and let $S\subseteq \R^d$.
\begin{definition}[Inference]
{We say that {$S$ infers $h$ at $x$}
if $\sign(\ip{h,x})$ is determined by
the linear queries $\sign(\ip{h',x})$ for $h' \in S$.
That is, if for any point $y$ in the set
$$
\left\{y \in \R^d: \sign(\ip{h',y})=\sign(\ip{h',x}) \quad \forall h' \in S\right\}
$$
it holds that $\sign(\ip{h,y})=\sign(\ip{h,x})$.}
Define
$$\inference(S;x) := \{h\in\R^d : \text{$h$ is inferred from $S$ at $x$}\}.$$

\end{definition}
The notion of inference has a natural geometric perspective. Consider the partition of $\R^d$ induced by $S$.
Then, $S$ infers $h$ at $x$ if the cell in this partition that contains $x$
is either disjoint from $h$ or otherwise is contained in $h$ (so in either case, the
value of $\sign(\ip{h,\cdot})$ is constant on the cell).

Our algorithms and analysis are based on inferences from comparisons.
Let $S-S$ denote the set $\{h'-h'': h',h'' \in S\}$.
\begin{definition}[Inference by comparisons]
{We say that \emph{comparisons on $S$} infer $h$ at $x$ if $S \cup (S-S)$
infers $h$ at $x$.} Define
$$\infer(S;x) := \inference\bigl(S\cup(S-S);x\bigr).$$
\end{definition}
Thus, $\infer(S;x)$ is determined by querying $\sign(\ip{h',x})$ and $\sign(\ip{h'-h'',x})$ for all $h',h'' \in S$.
Naively, this requires some $O(|S|^2)$ linear queries.
However, using efficient sorting algorithm (e.g.\ merge-sort)
achieves it with just  $O(|S|\log|S|)$ comparison queries.
A further improvement, when $|S| > d$, is obtained by Fredman's sorting algorithm
that uses just  $O(|S| + d\log |S|)$ comparison queries~\cite{fredman1976good}.

\subsection{Vectors in isotropic position}

Vectors $h_1,\ldots,h_m \in \R^d$ are said to be in \emph{general position} if any $d$ of them are linearly independent.
They are said to be in \emph{isotropic position} if for any unit vectors $v \in S^d$,
\[
\frac{1}{m} \sum_{i=1}^m \ip{h_i,v}^{2} = \frac{1}{d}.
\]
Equivalently, if $\tfrac{1}{m}\sum h_i h_i^T$ is $\tfrac{1}{d}$ times the $d \times d$ identity matrix.
An important theorem of Forster~\cite{forster02linear} (see also Barthe~\cite{barthe1998reverse} for a more general statement) states
that any set of vectors in general position can be scaled to be in isotropic position.

\begin{theorem}[\cite{forster02linear}]\label{thm:forster}
Let $H \subset \R^d$ be a finite set in general position.
Then there exists an invertible linear transformation $T$ such that the set
\[H':=\left\{\frac{Th}{\|Th\|_2} : h\in H\right\}\]
is in isotropic position. We refer to such a $T$ as a \emph{Forster transformation} for $H$.

\end{theorem}

We will also need a relaxed notion of isotropic position. Given vectors $h_1,\ldots,h_m \in \R^d$ and some $0<c<1$, we say that
the vectors are in \emph{$c$-approximate isotropic position}, if for all unit vectors $v \in S^d$ it holds that
\[
\frac{1}{m} \sum_{i=1}^m \ip{h_i,v}^2 \ge \frac{c}{d}.
\]
We note that this condition is easy to test algorithmically, as it is equivalent to the statement that the smallest eigenvalue of the positive semi-definite $d \times d$ matrix $\frac{1}{m} \sum_{i=1}^m h_i h_i^T$ is at least~$\frac{c}{d}$.

We summarize it in the following lemma, which follows from basic real linear algebra.
\begin{claim}\label{c:isotropiceigen}
Let $h_1,\ldots,h_m\in\R^d$ be unit vectors. Then the following are equivalent.
\begin{itemize}
\item $h_1,\ldots,h_m$ are in $c$-approximate isotropic position.
\item  $\lambda_1\left(\frac{1}{m}\sum_{i=1}^m h_i h_i^T\right)\geq c/d$,
\end{itemize}
where $\lambda_1(M)$ denotes the minimal eigenvalue of a positive
semidefinite matrix $M$.
\end{claim}

We will need the following basic claims. The first claim shows that a set of unit vectors in an approximate isotropic position
has many vectors with non-negligible inner product with any unit vector.

\begin{claim}\label{c:manymargin}
Let $h_1,\ldots,h_m\in\R^d$ be unit vectors in a $c$-approximate isotropic position,
and let $x\in\R^d$ be a unit vector.
Then, at least a $\frac{c}{2d}$-fraction of the $h_i$'s satisfy $\lvert \ip{h_i,x}\rvert > \sqrt{\frac{c}{2d}}$.
\end{claim}
\begin{proof}
Assume otherwise. It follows that
\[
\frac{1}{m}\sum_{i=1}^{m}\lvert \ip{h,x_i}\rvert^2
\leq
\frac{c}{2d}\cdot 1 + \left(1-\frac{c}{2d}\right)\frac{c}{2d}
<
\frac{c}{2d} + \frac{c}{2d}
=
\frac{c}{d}.
\]
This contradicts the assumption that the $h_i$'s are in $c$-approximate isotropic position.
\end{proof}

The second claim shows that a random subset of a set of unit vectors in an approximate isotropic position
is also in approximate isotropic position, with good probability.

\begin{claim}\label{c:chern}
Let $h_1,\ldots, h_m$ be unit vectors in $c$-approximate isotropic position.
Let $i_1,\ldots,i_k \in [m]$ be independently and uniformly sampled.
Then for any $\delta>0$, the vectors $h_{i_1},\ldots, h_{i_k}$ are in $\left((1-\delta) c\right)$-approximate isotropic position with probability at least
\[
1-d\cdot \left[ \frac{e^{-\delta}}{(1-\delta)^{1-\delta}}\right]^{ck/d}.
\]
\end{claim}

\begin{proof}
This is an immediate corollary of Matrix Chernoff bounds~\cite{tropp12friendly}.
By \Cref{c:isotropiceigen} the above event is equivalent to
that $\lambda_1\left(\frac{1}{k}\sum_{i=1}^{k}h_ih_i^T\right) \geq (1-\delta)\frac{c}{d}$.
By assumption, $\lambda_1\left(\frac{1}{m}\sum_{i=1}^mh_ih_i^T \right)\geq \frac{c}{d}$.
Now, by the Matrix Chernoff bound, for any $\delta \in [0,1]$ it holds that
\[
\Pr\left[ \lambda_1\left(\frac{1}{k}\sum_{i=1}^{k}h_ih_i^T\right) \leq (1-\delta)\cdot\frac{c}{d} \right]\leq
d\cdot \left[ \frac{e^{-\delta}}{(1-\delta)^{1-\delta}}\right]^{ck/d}.
\]
\end{proof}
We will use two instantiations of \Cref{c:chern}:
(i) $c\geq 3/4$, and $(1-\delta)c=1/2$,
and (ii) $c=1$ and $(1-\delta)c = 3/4$.
In both cases the bound simplifies to
\begin{equation}\label{eq:chern}
1-d\cdot \left(\frac{99}{100}\right)^{k/d}.
\end{equation}


\section{Proof of main theorem}

Let $H \subset \R^d$. We prove Theorem~\ref{thm:main} in four steps:
\begin{enumerate}
\item First, we assume that $H$ is in general position. In this case, we construct a randomized generalized comparison LDT
which computes $\A_H(\cdot)$, whose expected depth is $O(d^3 \log d \log |H|)$ for any input.
This is achieved in \Cref{sec:proof1}, see \Cref{lem:genposition}.

\item Next, we derandomize the construction. This gives
for any $H$ in general position a (deterministic) generalized comparison LDT which computes $\A_H(\cdot)$, whose depth is $O(d^4 \log d \log |H|)$.
This is achieved in \Cref{sec:proof2}, see \Cref{lem:genposdet}.

\item In the next step, we handle an arbitrary $H$ (not necessarily in general position), and construct by a compactness argument
a generalized comparisons LDT of depth $O(d^4 \log d \log |H|)$ which computes it almost everywhere.
This is achieved in \Cref{sec:proof3}, see \Cref{lem:almosteverywhere}.

\item Finally, we show that any LDT which computes $\A_H(\cdot)$ almost everywhere can be ``fixed" to a LDT
which computes $\A_H(\cdot)$ everywhere. This fixing procedure maintains both the depth of the LDT, as well as the set of queries
performed by it. This is achieved in \Cref{sec:proof4}, see \Cref{lem:LDT_fix}.
\end{enumerate}

\subsection{A randomized LDT for $H$ in general position}
\label{sec:proof1}

In this section we construct a randomized generalized comparison LDT for $H$ in general position. Here, by a randomized LDT we mean a distribution over (deterministic) LDT which
compute $\A_H(\cdot)$. The corresponding complexity measure is the expected number of queries it makes, for the worst-case input $x$.

\begin{lemma}\label{lem:genposition}
Let $H\subseteq\R^d$ be a finite set in general position.
Then there exists a randomized LDT that computes $\A_H(\cdot)$, which makes $O\left(d^3 \log d \log \lvert H\rvert\right)$ generalized comparison queries on expectation, for any input.
\end{lemma}

The proof of \Cref{lem:genposition} is based on a variant of the margin-based analysis of the inference dimension
with respect to comparison queries as in~\cite{kane2017active}
(The analysis in \cite{kane2017active} assumed that all vectors have large margin,
where here we need to work under the weaker assumption that only a noticeable fraction of the vectors have large margin).
The crux of the proof relies on scaling every $h\in H$
by a carefully chosen scalar $\alpha_h$ such that
drawing a sufficiently large random subset of $H$,
and sorting the values $\ip{\alpha_h h,x}$ using comparison queries (which correspond to generalized comparisons on the $h$'s)
allows to infer, on average, at least $\Omega(1/d)$ of the labels of $H$.
The scalars $\alpha_h$ are derived via Forster's theorem (\Cref{thm:forster}).
More specifically, $\alpha_h = \frac{1}{\|Th\|_2}$, where $T$ is a Forster transformation for $H$.

%
%

\begin{tcolorbox}
\begin{center}
{\bf Randomized generalized-comparisons tree for $H$ in general position}\\
\end{center}
\noindent
Let $H \subseteq \R^d$ in general position.\\
\\
Input: $x\in\R^d$, given by oracle access for $\sign(\ip{\cdot,x})$\\
Output: $\A_H(x) = \left(\sign(\ip{h,x})\right)_{h\in H}$
\begin{enumerate}
\item[(1)] Initialize: $H_0=H$, $i=0$, $v(h)=?$ for all $h \in H$. Set $k=\Theta(d^2\log(d))$.
\item[(2)] Repeat while $|H_i| \ge k$:
\begin{enumerate}
\item[(2.1)] Let $T_i$ be the Forster transformation for $H_i$. Define $H'_i = \left\{\frac{h}{\|T_i h\|_2}: h \in H_i\right\}$.
\item[(2.2)] Sample uniformly $S_i \subset H'_i$ of size $|S_i|=k$.
\item[(2.3)] Query $\sign(\ip{h,x})$ for $h\in S_i$ (using label queries).
\item[(2.4)] Sort $\ip{h,x}$ and $\ip{-h,x}$ for $h\in S_i$ (using generalized comparison queries).
\item[(2.5)] For all $h\in  H_i$, check if $h\in\infer(\pm S_i;x)$,
and in case it is, set $v(h) \in \{-,0,+\}$ to be the inferred value of $h$.
\item[(2.6)] Remove all $h\in H_i$ for which $\sign\left(\ip{h,x}\right)$ was inferred,
set $H_{i+1}$ to be the resulting set and go to step (2).
\end{enumerate}
\item[(3)] Query $\sign(\ip{h,x})$ for all $h \in H_i$, and set $v(h)$ accordingly.
\item[(4)] Return $v$ as the value of $\A_{H}(x)$.
\end{enumerate}
\end{tcolorbox}

In order to understand the intuition behind the main iteration (2) of the algorithm, define $x'=(T_i^{-1})^{\mathrm{T}} x$ and for each $h \in H_i$ let $h'=\frac{T_i h}{\|T_i h\|}$.
Then $\sign(\ip{h,x})=\sign(\ip{h',x'})$, and so it suffices to infer the sign for many $h' \in H_i$ with respect to $x'$. The main benefit is that
we may assume in the analysis that the set of vectors $H'_i$ is in isotropic position; and reduce the analysis to that of using (standard) comparisons
on $H'_i$ and $x'$. These then translate to performing generalized comparison queries on $H_i$ and the original input $x$. The following lemma
captures the analysis of the main iteration of the algorithm. Below, we denote by~$\pm S := S \cup (-S)$.

\begin{lemma}\label{lem:mainstep}
Let $x\in \R^d$, let $H \subseteq\R^d$ be a finite set of unit vectors in $c$-approximate isotropic position with $c\geq 3/4$,
and let $S \subset H$ be a uniformly chosen subset of size $k =  \Omega\left(d^2\log d \right)$. Then
\[\E_S\left[\lvert \infer(\pm S;x)\cap H \rvert\right]\geq \frac{\lvert H\rvert}{40d}.\]
\end{lemma}

Note that this proves a stronger statement than needed for \Cref{lem:genposition}.
Indeed, it would suffice to consider only $H$ that is in (a complete) isotropic position.
This stronger version will be used in the next section for derandomizing the above algorithm.
Let us first argue how \Cref{lem:genposition} follows from \Cref{lem:mainstep}, and then proceed
to prove \Cref{lem:mainstep}.

\begin{proof}[Proof of \Cref{lem:genposition} given \Cref{lem:mainstep}]
By \Cref{lem:mainstep}, in each iteration (2) of the algorithm, we infer on expectation at least $\Omega(1/d)$ fraction of the $h \in H'_i$
with respect to $x'=T_i^{-1} x$. By the discussion above, this is the same as inferring an $\Omega(1/d)$
fraction of the $h_i \in H_i$ with respect to $x$. So, the total expected number of iterations needed is $O(d \log |H|)$.
Next, we calculate the number of linear queries performed at each iteration.
The number of label queries is $O(k)$ and the number of comparison queries on $H'_i$ (which translate to generalized comparison
queries on $H_i$) is $O(k \log k)$ if we use merge-sort, and can be improved to $O(k + d \log k)$ by using Fredman's sorting algorithm~\cite{fredman1976good}.
So, in each iteration we perform $O(d^2 \log d)$ queries, and the expected number of iterations is $O(d \log |H|)$.
So the expected total number of queries by the algorithm is $O(d^3 \log d \log |H|)$.
\end{proof}


From now on, we focus on proving \Cref{lem:mainstep}. To this end, we assume from now that $H \subset \R^d$ is in $c$-isotropic position for $c \ge 3/4$.
Note that $h$ is inferred from comparisons on $\pm S$ if and only if $-h$ is, and that replacing an element of $S$ with its negation does not affect $\pm S$. Therefore, negating elements of $H$ does not change the expected number of elements inferred from comparisons on $\pm S$. Therefore, we may assume in the analysis that $\ip{h,x}\geq 0$ for all $h\in H$. Under this assumption, we will show that
\[\E_S\left[\lvert \infer(S;x)\cap H \rvert\right]\geq \frac{\lvert H\rvert}{40d}.\]
It is  convenient to analyze the following procedure for sampling $S$:
\begin{itemize}
\item Sample $h_1,\ldots h_{k+1}$ random points in $H$, and $r \in [k+1]$ uniformly at random.
\item Set $S =\{h_j : j \in [k+1] \setminus \{r\} \}$.
\end{itemize}
We will analyze the probability that comparisons on $S$ infer $h_r$ at $x$.
Our proof relies on the following observation.
\begin{obs}
The probability, according to the above process, that $h_r\in\infer(S;x)$
is equal to the expected fraction of $h\in H$ whose label is inferred. That is,
\[\Pr\left[h_r\in\infer(S;x)\right] = \E\left[\frac{\lvert\infer(S;x)\cap H\rvert}{\lvert H\rvert}\right].\]
\end{obs}
Thus, it suffices to show that $\Pr\left[h_r\in\infer(S;x)\right]\geq 1/40d$.
This is achieved by the next two propositions as follows.
\Cref{prop:isotropicsample} shows that $S$ is in a
$(1/2)$-approximate isotropic position with probability at least $1/2$,
and \Cref{prop:inference0} shows that whenever $S$ is in (1/2)-approximate isotropic position
then $h_r\in \infer(S;x)$ with probability at least~$1/20d$.
Combining these two propositions together yields that
$\Pr\left[h_r\in\infer(S;x)\right]\geq 1/40d$
and finishes the proof of Lemma~\ref{lem:mainstep}.

\begin{proposition}\label{prop:isotropicsample}
Let $H \subset \R^d$ be a set of unit vectors in $c$-approximate isotropic position for $c\geq3/4$.
Let $S \subset H$ be a uniformly sampled subset of size $|S| \ge \Omega(d\log d)$.
Then $S$ is in $(1/2)$-approximate isotropic position with probability at least $1/2$.
\end{proposition}

\begin{proof}
The proof follows from \Cref{c:chern} by plugging  $k=\Omega(d \log d)$ in \Cref{eq:chern} and calculating that the bound on the right hand side becomes at least~$1/2$.
\end{proof}

\begin{proposition}\label{prop:inference0}
Let $x \in \R^d$, $S \subset \R^d$ be in (1/2)-approximate isotropic position, where $|S| \ge \Omega\left(d^2\log d\right)$.
Let $h \in S$ be sampled uniformly. Then
\[\Pr_{h \in S}\left[h\in\infer\left(S \setminus \{h\} ; x\right)\right]\geq\frac{1}{20d}.\]
\end{proposition}

\begin{proof}
We may assume that $x$ is a unit vector, namely $\|x\|_2=1$. Let $s=|S|$ and assume that $S=\{h_1,\ldots,h_s\}$ with
\[
\ip{h_1,x} \ge \ip{h_2,x} \ge \ldots \ge \ip{h_s,x} \geq 0.
\]
Set $\eps=\frac{1}{2 \sqrt{d}}$. As $S$ is in (1/2)-approximate isotropic position, \Cref{c:manymargin} gives that $\ip{h_i,x} \ge \eps$ for
at least $|S|/4d$ many $h_i \in S$. Set $t=|S|/8d$ and define
\[
T=\{h_1,\ldots,h_t\},
\]
where by out assumption $\ip{h_t,x} \ge \eps$. Note that in this case, we can compute $T$ from comparison queries on $S$. We will show that
\[
\Pr_{h \in T}\left[h\in\infer\left(S \setminus \{h\} ; x\right)\right]\geq\frac{1}{2},
\]
from which the proposition follows. This in turn follows by the following two claims, whose proof we present shortly.

\begin{claim}\label{c:inference1}
Let $h_a \in T$. Assume that there exists a non-negative linear combination $v$ of $\{h_{i} - h_{i+1}: i=1,\ldots,a-2\}$ such that
$$
\| h_a - (h_1 + v)\|_2 \le \eps/4.
$$
Then $h_a \in \infer\left(S \setminus \{h_a\} ; x\right)$.
\end{claim}

\begin{claim}\label{c:inference2}
The assumption of \Cref{c:inference1} holds for at least half the vectors in $T$.
\end{claim}

Clearly, \Cref{c:inference1} and \Cref{c:inference2} together imply that for at least half of $h_a \in T$, it holds that
$h_a \in \infer\left(S \setminus \{h_a\} ; x\right)$. This concludes the proof of the proposition.
\end{proof}

Next we prove \Cref{c:inference1} and \Cref{c:inference2}.

\begin{proof}[Proof of \Cref{c:inference1}]
Let $S'=S \setminus \{h_a\}$ and $T'=T \setminus \{h_a\}$.
As $S$ is in (1/2)-approximate isotropic position then $S'$ is in $c$-approximate isotropic position for $c=1/2 - d/|S|$. In particular,
as $|S| \ge 4d$ we have $c \ge 1/4$. By applying comparison queries to $S'$ we can sort $\{\ip{h_i,x}: h_i \in S'\}$. Then $T'$ can be computed as the set
of the $t-1$ elements with the largest inner product. \Cref{c:manymargin} applied to $S'$ then implies that $\ip{h_i, x} \ge \eps/2$ for all $h_i \in T'$.
Crucially, we can deduce this just from the comparison queries on $S'$, together with our initial assumption that $S$ is in (1/2)-approximate isotropic position.
Thus we deduced from our queries that:
\begin{itemize}
\item $\ip{h_1,x} \ge \eps/2$.
\item $\ip{v,x} \ge 0$.
\end{itemize}
In addition, from our assumption it follows that $|\ip{h_a - (h_1+v),x}| \le \eps/4$. These together infer that $\ip{h_a,x}>0$.
\end{proof}

The proof of \Cref{c:inference2} follows from the applying the following claim iteratively. We note that this claim appears in \cite{kane2017active}
implicitly, but we repeat it here for clarity.

\begin{claim}\label{c:inference3}
Let $h_1,\ldots,h_t \in \R^d$ be unit vectors. For any $\eps>0$, if $t \ge 16d \ln (2d/\eps)$ then there exist $a \in [t]$
and $\alpha_1,\ldots,\alpha_{a-2} \in \{0,1,2\}$ such that
\[
h_{a} = h_1 + \sum_{j=1}^{i-2} \alpha_j (h_{j+1} - h_j) + e,
\]
where $\|e\|_2\leq \eps$.
\end{claim}

In order to derive \Cref{c:inference2} from \Cref{c:inference3}, we assume that $|T| \ge 32 d \ln ((2d)/(\eps/4)) = \Omega(d \log d)$. Then
we can apply \Cref{c:inference3} iteratively $|T|/2$ times with parameter $\eps/4$, at each step identify the required $h_a$, remove
it from $T$ and continue. Next we prove \Cref{c:inference3}.

\begin{proof}[Proof of \Cref{c:inference3}]
Let $\B:=\{h \in \R^d: \|h\|_2 \le 1\}$ denote the Euclidean ball of radius $1$,
and let $C$ denote the convex hull of $\{h_2-h_1,\ldots,h_t-h_{t-1}\}$.
Observe that $C \subset 2 \B$, as each $h_i$ is a unit vector.
For $\beta \in \{0,1\}^{t-1}$ define
\[
h_{\beta} = \sum \beta_j (h_{j+1}-h_j).
\]
We claim that having $t \geq 16d \ln (2d/\eps)$ guarantees that there exist distinct $\beta',\beta''$ for which
\[
h_{\beta'} - h_{\beta''} \in \frac{\eps}{4} (C-C).
\]
This follows by a packing argument: if not, then the sets $h_{\beta}+\frac{\eps}{4} C$ for $\beta\in\{0,1\}^{t-1}$ are mutually disjoint.
Each has volume $(\eps/4)^d \vol(C)$, and they are all contained in $t C$ which has volume $t^d \vol(C)$.
As the number of distinct $\beta$ is $2^{t-1}$ we obtain that
$2^{t-1} (\eps/4)^d \le t^d$, which contradicts our assumption on $t$.

Let $i \in [t]$ be maximal such that $\beta'_{i-1} \ne \beta''_{i-1}$.
We may assume without loss of generality that $\beta'_{i-1} = 0, \beta''_{i-1}=1$, as otherwise we can swap the roles of $\beta'$ and $\beta''$. Thus we have
$$
\sum_{j=1}^{i-1} (\beta'_j - \beta''_j) (h_{j+1} - h_j) \in \frac{\eps}{4} (C-C) \subset \eps \B.
$$
Adding $h_{i} - h_1 = \sum_{j=1}^{i-1} (h_{j+1}-h_j)$ to both sides gives
$$
\sum_{j=1}^{i-1} (\beta'_j - \beta''_j+1) (h_{j+1} - h_j) \in h_{i} - h_1 + \eps \B,
$$
which is equivalent to
$$
h_{i} - h_1 \in \sum_{j=1}^{i-1} (\beta'_j - \beta''_j+1) (h_{j+1} - h_j) + \eps \B.
$$
The claim follows by setting $\alpha_j = \beta'_j - \beta''_j+1$ and noting that by our construction $\alpha_{i-1}=0$, and hence the sum terminates at $i-2$.
\end{proof}

\subsection{A deterministic LDT for $H$ in general position}
\label{sec:proof2}

In this section, we derandomize the algorithm from the previous section.
We still assume that $H$ is in general position,
this assumption will be removed in the next sections.
\begin{lemma}\label{lem:genposdet}
Let $H\subseteq\R^d$ be a finite set in general position.
Then there exists an LDT that computes $\A_H(\cdot)$
with $O\left(d^4 \log d \log \lvert H\rvert\right)$ generalized comparison queries.
\end{lemma}

Note that the this bound is worse by a factor of $d$ than the one in \Cref{lem:genposition}. In \Cref{open:derand}
we ask whether this loss is necessary, or whether it can be avoided by a different derandomization technique.

 \Cref{lem:genposdet} follows by derandomizing the algorithm from \Cref{lem:genposition}.
Recall that \Cref{lem:genposition} boils down to showing that
$h\in\infer(S_i;x)$ for an $\Omega(1/d)$ fraction of $h\in H_i$ on average.
In other words, for every input vector $x$, \emph{most of the subsets $S_i\subseteq H'_i$} of size $\Omega(d^2 \log d)$
allow to infer from comparisons the labels of some $\Omega(1/d)$-fraction of the points in $H_i$.
We derandomize this step by showing that
\emph{there exists a universal set $S_i\subseteq H'_i$} of size $O(d^3 \log d)$
that allows to infer the labels of some $\Omega(1/d)$-fraction of the points in $H_i$,
with respect to \emph{any $x$}. This is achieved by the next lemma.
%

\begin{lemma}\label{lem:derand}
Let $H\subseteq \R^d$ be a set of unit vectors in isotropic position.
Then there exists~$S\subseteq H$ of size $O(d^3\log d)$
such that
\[\left(\forall x\in\R^d\right): \left\lvert \infer(S;x)\cap H \right\rvert \geq \frac{\lvert H\rvert}{100d}.\]
\end{lemma}

\begin{proof}
We use a variant of the double-sampling argument due to~\cite{vapnik1971uniform}
to show that a random $S\subseteq H$ of size $s=O(d^3\log d)$ satisfies the requirements.
Let $S=\{h_1,\ldots,h_s\}$ be a random (multi-)subset of size $s$,
and let $E=E(S)$ denote the event
\[E(S) := \left[ \exists x\in\R^d : \left\lvert\infer(S;x)\cap H\right\rvert < \lvert H\rvert/100d\right].\]
Our goal is showing that $\Pr[E] < 1$.
To this end we introduce an auxiliary event $F$.
Let $t=\Theta(d^2\log d)$, and let $T=\{h_1,\ldots,h_t\}\subseteq S$ be a subsample of $S$,
where each $h_i$ is drawn uniformly from $S$ and independently of the others.
Define $F=F(S,T)$ to be the event
\begin{align*}
F(S,T) := \Bigl[ \exists x\in\R^d :&\left\lvert\infer(T;x)\cap H\right\rvert < \lvert H\rvert/100d \text{ and}\\
 &\left\lvert\infer(T;x)\cap S\right\rvert \geq \lvert S\rvert/50d
  \Bigr].
\end{align*}
The following claims conclude the proof of \Cref{lem:derand}.
\begin{claim}\label{c:uc1}
If $\Pr[E] \ge 9/10$ then $\Pr[F] \ge 1/200 d$.
\end{claim}

\begin{claim}\label{c:uc2}
$\Pr[F] \le 1/250d$.
\end{claim}

This concludes the proof, as it shows that $\Pr[E] < 9/10$. We next move to prove \Cref{c:uc1} and \Cref{c:uc2}.

\begin{proof}[Proof of \Cref{c:uc1}]
Assume that $\Pr[E] \ge 9/10$. Define another auxiliary event $G=G(S)$ as
\[G(S) := \left[ \text{$S$ is in (3/4)-approximate isotropic position} \right].\]
Applying \Cref{c:chern} by plugging $m\geq 100d\ln(10d)$ in \Cref{eq:chern} gives that $\Pr[G] \ge 9/10$, which implies
that $\Pr[E \wedge G] \ge 8/10$. Next, we analyze $\Pr[F | E \wedge G]$.

To this end, fix $S$ such that both $E(S)$ and $G(S)$ hold. That is: $S$ is in $(3/4)$-approximate isotropic position, and there
exists $x=x(S) \in \R^d$ such that $|\infer(S;x) \cap H| < |H|/100d$. If we now sample $T \subset S$,
in order for $F(S,T)$ to hold, we need that (i) $|\infer(T;x) \cap H| < |H|/100d$ , which holds with probability one,
as $|\infer(S;x) \cap H| < |H|/100d$; and (ii) that $|\infer(T;x) \cap S| \ge |S|/50d$. So, we analyze this event next.

Applying \Cref{lem:mainstep} to the subsample $T$ with respect to $S$ gives that
$$
\E_T\left[|\infer(T;x) \cap S|\right] \ge |S|/40d.
$$
This then implies that
$$
\Pr\left[|\infer(T;x) \cap S| \ge |S|/100d \right] \ge 1/100d.
$$

To conclude: we proved under the assumptions of the lemma that $\Pr_S[E(S) \wedge G(S)] \ge 8/10$; and that for every $S$ which satisfies $E(S) \wedge G(S)$
it holds that $\Pr_T[F(S,T) | S] \ge 1/100d$. Together these give that $\Pr[F(S,T)] \ge 1/200 d$.
\end{proof}

\begin{proof}[Proof of \Cref{c:uc2}]
We can model the choice of $(S,T)$ as first sampling $T \subset H$ of size $t$, and then sampling $S \setminus T \subset H$ of size $s-t$.
We will prove the following (stronger) statement: for any choice of $T$,
\[\Pr \left[F(S,T)|T\right] < 1/250d.\]
So from now on, fix $T$ and consider the random choice of $T'=S \setminus T$. We want to show that:
\begin{align*}
\Pr_{T'}\Bigl[ (\exists x\in\R^d) : &\left\lvert\infer(T;x)\cap H\right\rvert < \lvert H\rvert/100d \text{ and }\\
&\left\lvert\infer(T;x)\cap S\right\rvert \geq \lvert S\rvert/50d.\Bigr]\leq 1/250d.
\end{align*}
We would like to prove this statement by applying a union bound over all $x\in \R^d$.
However, $\R^d$ is an infinite set and therefore a naive union seems problematic.
To this end we introduce a suitable equivalence relation
that is based on the following observation.
\begin{obs}
$\infer(T;x)$ is determined by $\sign(\ip{h,x})$ for $h \in T \cup (T-T)$.
\end{obs}
We thus define an equivalence relation on $\R^d$
where $x\sim y$ if and only if $\sign(\ip{h,x})=\sign(\ip{h,y})$ for all $h \in T \cup (T-T)$.
Let $C$ be a set of representatives for this relation.
Thus, it suffices to show that
\begin{align*}
\Pr_{T'}\Bigl[ (\exists x\in C) : &\left\lvert\infer(T;x)\cap H\right\rvert < \lvert H\rvert/100d \text{ and }\\
&\left\lvert\infer(T;x)\cap S\right\rvert \geq \lvert S\rvert/50d.\Bigr]\leq 1/250d.
\end{align*}
Since $C$ is finite, a union bound is now applicable.
Sepcifically, it is enough to show that
\begin{align*}
(\forall x\in C) : \Pr_{T'}\Bigl [&\left\lvert\infer(T;x)\cap H\right\rvert < \lvert H\rvert/100d \text{ and }\\
&\left\lvert\infer(T;x)\cap S\right\rvert \geq \lvert S\rvert/50d.\Bigr]\leq \frac{1}{250d|C|}.
\end{align*}
Now, (a variant of) Sauer's Lemma (see e.g.\ Lemma 2.1 in~\cite{kane2017near}) implies that
\begin{equation}
\lvert C \rvert\leq \left(2e\cdot\left\lvert T \cup (T-T) \right\rvert \right)^d\leq \left(2e \cdot t^2 \right)^d \leq (20t)^{2d}.
\end{equation}\label{eq:sauer}

Fix $x\in C$. If $\left\lvert\infer(T;x)\cap H\right\rvert \geq \frac{\lvert H\rvert}{100d}$ then we are done
(note that $\infer(T;x)$ is fixed since it depends only on $T$ and $x$ and not on $T'$).
So, we may assume that $\left\lvert\infer(T;x)\cap H\right\rvert < \frac{\lvert H\rvert}{100d}$. Then we need to bound
\[
\Pr\left[\left\lvert\infer(T;x)\cap S\right\rvert \geq \frac{\lvert S\rvert}{50d} \right]
\leq \Pr\left[\left\lvert\infer(T;x)\cap T'\right\rvert \geq \frac{\lvert T'\rvert}{75d} \right],
\]
where the inequality follows if $t\leq \frac{s}{150d}$,
which can be satisfied since $t=\Theta(d^2\log d)$ and $s=\Theta(d^3 \log d)$.
To bound this probability we use the Chernoff bound:
let $p=\frac{\lvert \infer(T;x)\cap H \rvert}{\lvert H\rvert}$;
note that $\lvert \infer(T;x)\cap T' \rvert$ is distributed like $\mathsf{Bin}(s-t,p)$.
By assumption, $p\leq\frac{1}{100d}$, and therefore:
\begin{align*}
\Pr\left[\left\lvert\infer(T;x)\cap T'\right\rvert \geq \frac{\lvert T'\rvert}{75d} \right]
\leq
\exp\left(-\frac{(1/3)^2\cdot(t/100d)}{3}\right)
\leq
\frac{1}{250d\cdot (20t)^{2d}}
\leq
\frac{1}{250d\cdot \lvert C\rvert},
\end{align*}
where the second inequality follows because $t=\Theta(d^2\ln(d))$ with a large enough constant,
and the last inequality follows by \Cref{eq:sauer}.
\end{proof}
\end{proof}

\subsection{An LDT for every $H$ that is correct almost everywhere}
\label{sec:proof3}

In the next two sections we extend the generalized comparison decision tree
to arbitrary sets $H$.
Let $H\subseteq \R^d$ be an arbitrary finite set (not necessarily in a general position).
The first step is to use a compactness argument
to derive a decision tree that computes $\A_H(x)$ for \emph{almost every} $x$ in the following sense.
Recall that a linear decision tree $T$ computes the function $x\to\sign\left(\ip{h,x}\right)$
almost everywhere if this function is constant on every full dimensional leaf of $T$.

\begin{lemma}\label{lem:almosteverywhere}
Let $H\subseteq\R^d$ be a finite set.
Then there exists a generalized comparison LDT
of depth $O\left(d^4 \log d \log \lvert H\rvert\right)$
that computes $\A_H(\cdot)$ almost everywhere.
\end{lemma}

\begin{proof}
If $H$ is in general position then this follows from Lemma~\ref{lem:genposdet}.
So, assume that $H$ is not in general position.
For every $n\in\N$, pick $H_n \subset \R^d$ with $\lvert H_n\rvert = \lvert H\rvert$
in general position such that for every $h\in H$ there exists $h_n=h_n(h)\in H_n$
with $\| h- h_n\|_2\leq 1/n$.
By Lemma~\ref{lem:genposdet} each $H_n$ has a generalized comparisons
tree $T_n$ of depth $D = O\left(d^4 \log d \log \lvert H\rvert\right)$ that computes $\A_{H_n}(\cdot)$.
A standard compactness\footnote{For the compactness argument to carry,  each generalized comparison query is normalized so that its coefficient $\alpha,\beta$ are bounded, e.g.\ $ \alpha + \beta  =1$}
argument imply the existence of a sequence of isomorphic trees $\{T_{n_k}\}_{k=1}^\infty$,
such that for every vertex $v$, the sequence of
the $H_{n_k}$-generalized comparisons queries corresponding
to $v$ converges to an $H$-generalized comparison query.
Let $T_\infty$ denote the limit tree.
One can verify that $T_\infty$ satisfies the following property:
\begin{equation}\label{eq:compact}
C_\infty(\ell)\subseteq \bigcup_{j=1}^\infty\bigcap_{k=j}^\infty C_{n_k}(\ell), \text{ for every full dimensional leaf $\ell$ of $T_\infty$.}
\end{equation}
In words, every $x\in C_\infty(\ell)$ belongs to all except finitely
many of the $C_{n_k}(\ell)$.
We claim that \Cref{eq:compact} implies that $T_\infty$
computes $\sign(\ip{h,\cdot})$ almost everywhere, for every $h\in H$.
Indeed, let $\ell$ be a full dimensional leaf of $T_\infty$, and let $x',x''\in C_\infty(\ell)$.
Assume towards contradiction that $\sign(\ip{h,x'})\neq\sign(\ip{h,x''})$ for some $h\in H$.
By Corollary~\ref{lem:all} we may assume that $\sign(\ip{h,x'}) = -1$
and $\sign(\ip{h,x''})=+1$.
By \Cref{eq:compact}, both $x',x''$ belong to all but finitely many of the $C_{n_k}(\ell)$.
Moreover, since $\sign(\ip{h,x'}),\sign(\ip{h,x''})\neq 0$ and $h_{n_k}(h)\to_{k\to\infty} h$ it follows that
$\sign(\ip{h,x'})=\sign(\ip{h_{n_k}(h),x'})$, and $\sign(\ip{h,x''})=\sign(\ip{h_{n_k}(h),x''})$
for all but finitely many~$k$'s. Thus, for such $k$'s the function $\sign(\ip{h_{n_k},\cdot})$
is not constant on $C_{n_k}(\ell)$, which contradicts the assumption that $T_{n_k}$
computes $\sign(\ip{h_{n_k},\cdot})$.
\end{proof}

\subsection{An LDT for every $H$}
\label{sec:proof4}

In this section we derive the final generalized comparison decision tree for arbitrary $H$,
which implies \Cref{thm:main}.
This is achieved by the next lemma that derives the final tree
from the one in \Cref{lem:almosteverywhere}.

\begin{lemma}
\label{lem:LDT_fix}
For every LDT $T$ there exists an LDT $T'$ such that
\begin{itemize}
\item $T'$  uses the same queries as $T$ and has the same depth as $T$.
\item For every $h\in\R^d$, if $T$ computes $\sign(\ip{h,\cdot})$ almost everywhere
then $T'$ computes $\sign(\ip{h,\cdot})$ everywhere.
\end{itemize}
\end{lemma}
\begin{proof}
Without loss of generality, we may assume that $T$ is not redundant, in the sense
that each query in it is informative. Namely, that $C(u)\neq\emptyset$ for every vertex $u\in T$.

\paragraph{Derivation of $T'$.}
Given an input point $x$, follow the corresponding computation path in $T$ with the following modification:
once a vertex $v$ whose  query $q=q(v)$ satisfies $\sign(\ip{q,x})=0$ is reached,
define in $T'$ a new child of $v$ that corresponds to this case, and continue following
the same queries like in the subtree of $T$ that corresponds to the case ``$\sign(\ip{q,x})=+$''.

\paragraph{Correctness.}
We prove that $T'$ computes $\sign(\ip{h,\cdot})$ everywhere.
Consider a leaf $\ell$ in $T'$, and let $x',x''\in C(\ell)$.
Assume toward contradiction that $\sign(\ip{h,x'})\neq \sign(\ip{h,x''})$.
By Corollary~\ref{lem:all} we may assume that $\sign(\ip{h,x'}) = -1$
and $\sign(\ip{h,x''})=+1$.
Let $q_1,\ldots,q_r$ denote the queries on the path towards $\ell$
whose query is replied by $0$. Since $T$ is not redundant,
it follows that the $q_i$'s are linearly independent.
Thus, there is a solution $z$ to the system $\ip{q_i,z} = +1$ for $1\leq i \leq r$.
Let $y' = x' +\eps\cdot z$ and $y''=x''+\eps\cdot z$ where $\eps>0$ is sufficiently small
such that
\begin{itemize}
\item[(i)] $\sign(\ip{h,x'})=-1$ and $\sign(\ip{h,x''})=+1$, and
\item[(ii)] $\sign(\ip{q,y'})=\sign(\ip{q,x'}) = \sign(\ip{q,x''})=\sign(\ip{q,y''})$
for every query $q$ on the path towards $\ell$ whose sign query is replied by a $\pm 1$.
\end{itemize}

Thus, by (ii) above and since $\eps>0$ it follows that $y',y''$ belong to the same full dimensional leaf of $T$.
However, (i) above implies that the function $x\mapsto \sign(\ip{h,x})$ is not constant on this leaf,
which contradicts the assumption on $T$.
\end{proof}

\bibliographystyle{alpha}
\bibliography{comparisons}

\end{document}